\documentclass[letterpaper, 10 pt, conference]{ieeeconf}  

\IEEEoverridecommandlockouts                              
\usepackage{graphicx} 
\usepackage{times} 
\usepackage{amsmath} 
\usepackage{amssymb}  
\usepackage{cleveref}
\usepackage{subcaption}
\usepackage{pgfplots}

\newcommand{\R}{\mathbb{R}}

\newcommand{\Lcal}{\mathcal{L}}

\DeclareMathOperator{\aslim}{as-lim}

\newtheorem{theorem}{Theorem}
\newtheorem{lemma}{Lemma}

\newtheorem{corollary}{Corollary}
\newtheorem{remark}{Remark}

\usepackage[noadjust]{cite}

\def\doubleunderline#1{\underline{\underline{#1}}}

\title{\LARGE \bf Covariance-Robust Dynamic Watermarking}

\author{Matt Olfat, Stephen Sloan, Pedro Hespanhol, Matt Porter, Ram Vasudevan, and Anil Aswani
\thanks{This material is based upon work partially supported by the National Science Foundation under Grant CMMI-1847666, by the UC Berkeley Center for Long-Term Cybersecurity, and by a grant from Ford Motor Company via the Ford-UM Alliance under award N022977.}
\thanks{Matt Olfat, Steve Sloan, Pedro Hespanhol, and Anil Aswani are with the Department of Industrial Engineering and Operations Research, University of California, Berkeley 94720 {\tt\small \{molfat,stephen\_sloan,pedrohespanhol,aaswani\}@}}%
\thanks{\hspace{-8pt}\tt\small berkeley.edu}
\thanks{Matt Porter and Ram Vasudevan are with the Department of Mechanical Engineering, University of Michigan, Ann Arbor 48109 {\tt\small \{matthepo,ramv\}@umich.edu}}}%
								
\begin{document}

\maketitle
\thispagestyle{empty}
\pagestyle{empty}

\begin{abstract}
Attack detection and mitigation strategies for cyberphysical systems (CPS) are an active area of research, and researchers have developed a variety of attack-detection tools such as dynamic watermarking. However, such methods often make assumptions that are difficult to guarantee, such as exact knowledge of the distribution of measurement noise. Here, we develop a new dynamic watermarking method that we call covariance-robust dynamic watermarking, which is able to handle uncertainties in the covariance of measurement noise. Specifically, we consider two cases. In the first this covariance is fixed but unknown, and in the second this covariance is slowly-varying. For our tests, we only require knowledge of a set within which the covariance lies. Furthermore, we connect this problem to that of algorithmic fairness and the nascent field of fair hypothesis testing, and we show that our tests satisfy some notions of fairness. Finally, we exhibit the efficacy of our tests on empirical examples chosen to reflect values observed in a standard simulation model of autonomous vehicles.
\end{abstract}

\section{INTRODUCTION}

The development of 5G, the ``fifth generation" of wireless technology, brings with it increased bandwidth, massive-scale device-to-device (D2D) connections, lower latency, and high reliability. The latency reductions with 5G open the door to further growth in cyberphysical systems (CPS), which involve the intercommunication and real-time management of large numbers of physical sensors and actuators, often in shifting environments \cite{ferrag2018security}. System vulnerabilities to malicious agents abound in all of these technologies \cite{abrams2008malicious,langner2011stuxnet,cardenas2008research}, and 5G in particular necessitates more robust cyber-security measures for the relevant control systems \cite{ferrag2018security}. 

Much of the existing work on security for CPS assumes that the system is fixed and all required distributions are exactly known \cite{satchidanandan2016dynamic,weerakkody2014detecting,mo2014detecting}. However, the system description is often time-varying or partially uncertain for many CPS \cite{zhang2016uplink,ahmad2019quality}.  Given this real-world motivation of security for time-varying or partially unknown CPS, we focus in this paper on designing robust security schemes to test for adversarial attacks on LTI systems. Recent work has established dynamic watermarking as a key active method for detecting sensor attacks \cite{satchidanandan2016dynamic,weerakkody2014detecting,satchidanandan2017minimal,mo2009secure,mo2010false,mo2014detecting,mo2015physical,porter2020detecting}, and here we build on this work by designing covariance-robust dynamic watermarking. 

We design robust watermarking for two sub-cases: The first is where the covariance of measurement noise is fixed but unknown, and the second is where the covariance of measurement noise is unknown and slowly-varying. The first reflects a scenario of many nearly-identical systems with variation between copies of the system. The second sub-case reflects a scenario where a sensor has different accuracy in varying regimes, such as lidar on an autonomous vehicle in changing weather. Attack detection is critical in all of these such cases, and we need statistical tests that retain their power in the face of system changes or uncertainty.


\subsection{Fairness}
\label{sec:fairness}

Robust data-driven decision-making has gained attention in the literature on algorithmic fairness. Motivated by machine learning tasks with societal applications, the fairness literature has sought to design learning methods that refrain from considering certain variables. To that extent, this body of work defines rigorous, mathematical notions of fairness for supervised learning \cite{berk2017fairness,hardt2016equality,calders2009building,dwork2012fairness,zliobaite2015relation,olfat2018spectral,chouldechova2017fair,zafar2017}, which have recently been extended to unsupervised learning by \cite{chierichetti2017fair,olfat2019convex}. 

The work in \cite{aswani2019optimization} outlines a general framework: Consider $(X,Y,Z)$ with a joint distribution $\mathbb{P}$, where $X$ are exogenous inputs, $Y$ are endogenous ``targets", and $Z$ is a ``protected attribute". The goal is to choose a \textit{decision rule} $\delta(x)$ that makes a decision $d$ using inputs $X$, in order to minimize some \textit{risk function} $\mathcal{R}_{\mathbb{P}}(\delta,Y)$. In dynamic watermarking: $X$ are measurements, $Y$ is a binary variable that denotes if the system is under attack, and $Z$ is the true system characterization; our decision rule $\delta$ for if the system is under attack is made without $Y$ and $Z$, which are not observed. We then define a decision rule to be without \textit{disparate impact} if
\begin{equation}
\label{eq:dispimp}
\delta^*\in\arg\min_{\delta}\big\{R_{\mathbb{P}}(\delta,Y)\ \big|\ \delta(X)\perp \!\!\! \perp Z\big\},
\end{equation}
where $\delta(X)\perp \!\!\! \perp Z$ means $\delta(X)$ is independent of $Z$. This increases fairness because it removes any impact of $Z$ on the decision by imposing independence as a constraint. However, some \cite{dwork2012fairness,hardt2016equality} have argued that this above definition of fairness can be too restrictive in some cases and that \textit{equalized odds} is a better definition of fairness. Its only difference is that in (\ref{eq:dispimp}) we replace $\delta(X)\perp \!\!\! \perp Z$ with $(\delta(X)\perp \!\!\! \perp Z) | Y$. That is, equalized odds ask for independence of $\delta(X)$ and $Z$ when conditioned on $Y$. We can interpret equalized odds as requiring error rates to be similar across protected groups. Finally, a notion associated equalized odds is that of \textit{equal opportunity}, which amounts to enforcing $(\delta(X)\perp \!\!\! \perp Z) | Y=y$, for some value $y$. This is relevant when one particular type of error is of more interest than another.

\subsection{Relevance of Fairness to Watermarking}

Fairness is relevant to the design of robust tests for two reasons. First, it provides a well-established technical language with which to discuss our requirement of robustness. Past dynamic watermarking techniques require exact system knowledge, and as such the corresponding watermarking tests will have error rates that are biased over inevitable system perturbations or uncertainties. Fairness notions such as \textit{equalized odds} and \textit{equal opportunity} allow for more specific framing of the problem and thus give a framework to design more robust methods for dynamic watermarking.

Second, robust cyber-security methods will have improved social impacts, which is the most general way of interpreting ``fairness". For example, smart homes can have many sensors. Changes in the distribution of sensor noise can correlate with factors such as climate, which correlates with geography and thus attributes like race, ethnicity, or class. A systemic bias in the ability to detect threats thus yields, and possibly perpetuates, systemic bias in outcomes among these groups. Robustness of cyber-security methods thus have the potential to improve societal fairness of the corresponding methods.

\subsection{Outline}

In Sect. \ref{sec:dwm}, we outline key terminology and results in dynamic watermarking. In Sect. \ref{sec:fixed}, we present our covariance-robust dynamic watermarking scheme for the case of fixed, but unknown, measurement noise covariance. This is then extended in Sect. \ref{sec:varying} to the case where measurement noise covariance is allowed to slowly vary. Sect. \ref{sec:results} presents empirical results that demonstrate efficacy of our approach.

\section{Preliminaries}

\label{sec:dwm}

We describe the LTI system and attack models, and then review existing results about dynamic watermarking.

\subsection{LTI System Model}

Consider a partially-observed MIMO LTI system
\begin{equation}
\begin{aligned}
x_{n+1} &= Ax_n+Bu_n+w_n\\
y_n &= Cx_n+z_n+v_n
\end{aligned}
\end{equation}
for $x_n,w_n\in\R^p, u_n\in\R^q$ and $y_n,z_n,v_n\in\R^m$. Here $w_n$ is mean-zero i.i.d. multivariate Gaussian process noise with covariance matrix $\Sigma_W$, and this is independent of $z_n$ that is i.i.d. Gaussian measurement noise with mean-zero; but we assume that the covariance matrix for $z_n$ is a linear function $\Sigma_Z(\theta)$ of a set of parameters $\theta\in\mathcal{P}\subset\R^d$ taking values in polyhedron $\mathcal{P}$. For now, $\theta$ is assumed constant but unknown for any fixed system. The $v_n$ is an additive signal chosen by an attacker who seeks to corrupt sensor measurements.

Stabilizability of $(A,B)$ and detectability of $(A,C)$ imply the existence of a controller $K$ and observer $L$ such that $A+BK$ and $A+LC$ are Schur stable. The closed-loop system can be stabilized using the control input $u_n = K\hat{x}_n$, where $\hat{x}_n$ is the observer-estimated state. Define $\tilde{x}_n^{\vphantom{\textsf{T}}} = \begin{bmatrix} x_n^{\textsf{T}} & \hat{x}_n^{\textsf{T}} \end{bmatrix}^{\textsf{T}}$, $\underline{D}^{\vphantom{\textsf{T}}} = \begin{bmatrix} I^{\vphantom{\textsf{T}}} & 0^{\vphantom{\textsf{T}}} \end{bmatrix}^{\textsf{T}}$, $\underline{L}^{\vphantom{\textsf{T}}} = \begin{bmatrix} 0^{\vphantom{\textsf{T}}} & -L^{\textsf{T}} \end{bmatrix}^{\textsf{T}}$, and
\begin{equation}
\underline{A} = \begin{bmatrix}
A & BK \\
-LC & A+BK+LC
\end{bmatrix}.
\end{equation}
We can write the closed-loop evolution of the state and estimated state when $v_n\equiv 0$ as $\tilde{x}_{n+1} = \underline{A}\tilde{x}_n + \underline{D}w_n + \underline{L}z_n$. Alternatively, we may define the observation error $\delta_n=\hat{x}_n-x_n$. Let $\breve{x}_n^{\vphantom{\textsf{T}}} = \begin{bmatrix} x_n^{\textsf{T}} & \delta_n^{\textsf{T}} \end{bmatrix}^{\textsf{T}}$, $\doubleunderline{D}^{\vphantom{\textsf{T}}} = \begin{bmatrix} I^{\vphantom{\textsf{T}}} & -I^{\vphantom{\textsf{T}}} \end{bmatrix}^{\textsf{T}}$, $\doubleunderline{L}^{\vphantom{\textsf{T}}} = \underline{L}^{\vphantom{\textsf{T}}}$, and
\begin{equation}
\label{eq:closedloop}
\doubleunderline{A} = \begin{bmatrix}
A + BK & BK \\
0 & A+LC
\end{bmatrix}.
\end{equation}
The closed-loop system for this change of variables is $\breve{x}_{n+1} = \doubleunderline{A}\breve{x}_n + \doubleunderline{D}w_n + \doubleunderline{L}z_n$. Note that $\doubleunderline{A}$ is Schur stable since both $A+BK$ and $A+LC$ are Schur stable.

\subsection{Attack Model}

Following \cite{hespanhol2017dynamic}, we consider attacks where $v_n=\alpha(Cx_n+z_n) + C\eta_n + \zeta_n$ for a fixed $\alpha\in\R$ and i.i.d. Gaussian $\zeta_n$ with mean-zero and covariance matrix $\Sigma_S$. Here, the $\eta_n$ are chosen to follow the process $\eta_{n+1} = (A+BK)\eta_n + \omega_n$, where $\omega_n$ are similarly i.i.d. Gaussian with mean-zero and covariance matrix $\Sigma_O$. The implication is that the attacker minimizes or mutes the true output $Cx_n+z_n$, and instead replaces it with a simulated output that follows the system dynamics and is thus not easily distinguishable as false. Furthermore, the attacker has access to process $w_n$ and measurement noise $z_n$. With this attack, the closed-loop systems above become $\tilde{x}_{n+1} = \underline{A}\tilde{x}_n + \underline{D}w_n + \underline{L}(z_n+v_n)$ and $\breve{x}_{n+1} = \doubleunderline{A}\breve{x}_n + \doubleunderline{D}w_n + \doubleunderline{L}(z_n+v_n)$.

\subsection{(Nonrobust) Dynamic Watermarking}

The steady-state distribution of $\delta_n$ in an unattacked system will be Gaussian with mean-zero and a covariance matrix of
\begin{equation}
\label{eq:sigmadelta}
\Sigma_{\Delta} = (A+LC)^{\vphantom{\textsf{T}}}\Sigma_{\Delta}^{\vphantom{\textsf{T}}}(A+LC)^{\textsf{T}} + \Sigma_W + L^{\vphantom{\textsf{T}}}\Sigma_Z(\theta)^{\vphantom{\textsf{T}}}L^{\textsf{T}}.
\end{equation}
Dynamic watermarking adds a small amount of Gaussian noise $e_n$, the values unknown to the attacker, into the control input $u_n=K\hat{x}_n+e_n$. This private excitation has mean-zero and covariance matrix $\Sigma_E$. Defining $\underline{B}^{\vphantom{\textsf{T}}} = \begin{bmatrix} B^{\textsf{T}} & B^{\textsf{T}} \end{bmatrix}^{\textsf{T}}$ and $\doubleunderline{B}^{\vphantom{\textsf{T}}} = \begin{bmatrix} B^{\textsf{T}} & 0^{\vphantom{\textsf{T}}} \end{bmatrix}^{\textsf{T}}$, the closed-loop systems with watermarking are given by $\tilde{x}_{n+1} = \underline{A}\tilde{x}_n + \underline{B}e_n + \underline{D}w_n + \underline{L}(z_n+v_n)$ and $\breve{x}_{n+1} = \doubleunderline{A}\breve{x}_t + \doubleunderline{B}e_n + \doubleunderline{D}w_n + \doubleunderline{L}(z_n+v_n)$, respectively.

The watermarking noise $e_n$ leaves a detectable signal in the measurements $y_n$, which can detect the presence of an attack $v_n$ by comparing the observer error $C\hat{x}_n-y_n$ to previous values of the watermark $e_{n-k}$ for some integer $k>0$. Specifically, the work in \cite{hespanhol2017dynamic} proposes the tests
\begin{align}
&\textstyle\aslim_{N\rightarrow\infty}\frac{1}{N}\sum_{n=0}^{N-1}(C\hat{x}_n-y_n)^{\vphantom{\textsf{T}}}(C\hat{x}_n-y_n)^{\textsf{T}} = \nonumber\\
&\hspace{5.5cm}C^{\vphantom{\textsf{T}}}\Sigma_{\Delta}^{\vphantom{\textsf{T}}}C^{\textsf{T}}+\Sigma_Z^{\vphantom{\textsf{T}}}\label{eq:test1}
\\
\label{eq:test2}
&\textstyle\aslim_{N\rightarrow\infty}\frac{1}{N}\sum_{n=0}^{N-1}(C\hat{x}_n-y_n)^{\vphantom{\textsf{T}}}e_{n-k'-1}^{\textsf{T}} = 0,
\end{align}
where $k'=\min_{k\ge1}\{C^{\vphantom{\textsf{T}}}(A+BK)^kB^{\textsf{T}}\ne0\}$. Any modeled attack passing these tests can be shown to asymptotically have zero power $\aslim_{N\rightarrow\infty}\frac{1}{N}\sum_{n=0}^{N-1}v_n^{\textsf{T}}v_n^{\vphantom{\textsf{T}}} = 0$ \cite{hespanhol2017dynamic}.  

Finally, \cite{hespanhol2017dynamic} also provides a test statistic for implementing the above test. Define $\psi_n = \begin{bmatrix} (C\hat{x}_n-y_n)^{\textsf{T}} & e_{n-k'-1}^{\textsf{T}} \end{bmatrix}^{\textsf{T}}$ and $S_n=\sum_{i=n+1}^{n+\ell}\psi_n^{\vphantom{\textsf{T}}}\psi_n^{\textsf{T}}$. Then the negative log-likelihood of a Wishart distribution is
\begin{equation}\tag{DW}
\label{eq:oldstatistic}
\begin{aligned}
\mathcal{L}=&(m+q+1-\ell)\log\det{S_n} \\
&+ \textrm{trace}\Bigg\lbrace
\begin{bmatrix}
\big(C^{\vphantom{\textsf{T}}}\Sigma_{\Delta}^{\vphantom{\textsf{T}}} C^{\textsf{T}}+\Sigma_Z\big)^{-1} & 0 \\
0 & \Sigma_E^{-1}
\end{bmatrix}\times S_n\Bigg\rbrace.
\end{aligned}
\end{equation}
This can be used to perform a statistical hypothesis test to detect attacks when using dynamic watermarking.

\section{Covariance-Robust Dynamic Watermarking}
\label{sec:crdw}

We develop covariance-robust dynamic watermarking methods for two different cases. The first is where $\theta$ is fixed but unknown, and the second is where $\theta$ is slowly varying.

\subsection{Fixed But Unknown Noise Covariance}
\label{sec:fixed}
We begin by stating our assumptions for this case. First, we assume that we have knowledge of a set of positive semidefinite matrices $\Sigma_{z,1},\dots,\Sigma_{z,d}$ such that these matrices are affinely independent and $\Sigma_Z(\theta) \in \mathrm{int}(\Omega^Z)$ for the set
\begin{equation}
\Omega^Z=\{\theta_1\Sigma_{z,1}+\cdots+\theta_d{\Sigma}_{z,d}:\mathbf{1}^T\theta=1, \theta\ge\mathbf{0}\}.
\end{equation}
Note that $\Omega^Z$ is a polyhedron, and that this set is defined to be the convex combination of ${\Sigma}_{z,1},\dots,{\Sigma}_{z,d}$. Our first result characterizes $\Omega^{\Delta}$, which is the set of possible $\Sigma_{\Delta}(\theta)$.

\begin{lemma}
	\label{lemma:sigmadelta}
	Let $\bar{\Sigma}_{\delta,k}$ satisfy $\bar{\Sigma}_{\delta,k}=(A+LC)\bar{\Sigma}_{\delta,k}(A+LC)^T+\Sigma_W+L{\Sigma}_{z,k}L^T$. For $\Sigma_{Z}(\theta)=\theta_{1}{\Sigma}_{z,1}+\cdots+\theta_{d}{\Sigma}_{z,d}$, the solution to (\ref{eq:sigmadelta}) is $\Sigma_{\Delta}(\theta)=\theta_{1}\bar{\Sigma}_{\delta,1}+\cdots+\theta_{d}\bar{\Sigma}_{\delta,d}$.
\end{lemma}

\begin{proof}
This immediately follows by noting that both sides of (\ref{eq:sigmadelta}) are linear in the matrices $\Sigma_\Delta$ and $\Sigma_Z(\theta)$.
\end{proof}

Since $E[\psi_n^{\vphantom{\textsf{T}}}\psi_n^{\textsf{T}}] = \mathrm{blkdiag}\{C^{\vphantom{\textsf{T}}}\Sigma_{\Delta}^{\vphantom{\textsf{T}}} C^{\textsf{T}}+\Sigma_Z,\Sigma_E\}$, we need to characterize the set $\Omega$ of feasible matrices in terms of $\theta$. 

\begin{lemma}
Let $\bar{\Sigma}_k = \mathrm{blkdiag}\{C^{\vphantom{\textsf{T}}}\bar{\Sigma}_{\delta,k}^{\vphantom{\textsf{T}}} C^{\textsf{T}}+\Sigma_{z,k}, \Sigma_E\}$. Then $\Omega=\{\theta_1\bar{\Sigma}_k+\cdots+\theta_d\bar{\Sigma}_d:\mathbf{1}^T\theta=1,\theta\ge\mathbf{0}\}$.
\end{lemma}

\begin{proof}
This follows by the linearity in $\Sigma_\Delta$ and $\Sigma_Z$.
\end{proof}

The set $\Omega$ represents covariance matrices of $\psi_n$ that are ``acceptable", according to the original set $\Omega^Z$ of observation noise covariances that we should not mistake for attacks.

\begin{lemma}
	\label{lemma:fulldim}
	The set $\Omega$ is of dimension $d-1$.
\end{lemma}

\begin{proof}
	This follows from Lemma \ref{lemma:sigmadelta}, the fact that $L$ is of full column-rank, and the observability of $(A+LC,C)$, which in turn follows from the observability of $(A,C)$.
\end{proof}

Finally, consider a modification of (\ref{eq:oldstatistic}) given by
\begin{multline}
\label{eq:loglikelihood}
\Lcal(S_n,V)=(m+q+1-\ell)\log\det{S_n} + \\\textrm{trace}\big\lbrace V S_n\big\rbrace- \ell\log\det{V}.
\end{multline}
Note (\ref{eq:loglikelihood}) is the negative log-likelihood of an $(m+q)\times (m+q)$ Wishart distribution with scale matrix $V^{-1}$ and $\ell$ degrees of freedom. Now, we may present our test statistic. Let $\Omega^{-1}=\{V: V^{-1}\in\Omega\}$ and define the test statistic
\begin{equation}
T(S_n)=\min_{V\in\Omega^{-1}}\Lcal(S_n,V)
\end{equation}
for the composite null hypothesis $H_0: E[\psi_n^{\vphantom{\textsf{T}}}\psi_n^{\textsf{T}}]\in\textrm{int}(\Omega)$. For some $0\le\nu$, consider the test
\begin{equation}
\label{eq:test3}
\begin{cases}
	\textrm{reject } H_0 & \textrm{if } T(S_n)>\nu\\
	\textrm{accept } H_0 & \textrm{if } T(S_n)\le\nu.
\end{cases}
\end{equation}
Since $\arg\min_{V\in\Omega^{-1}}\Lcal(S_n,V)=S_n^{-1}$, this proposed test is equivalent to the generalized likelihood ratio test.

\begin{theorem}
	For large enough $\ell$, the decision rule (\ref{eq:test3}) using test statistic $T(S_n)$ satisfies equal opportunity with respect to the null hypothesis and where the protected attribute is the true measurement noise covariance $\Sigma_Z(\theta)\in\mathrm{int}(\Omega^Z)$.
\end{theorem}

\begin{proof}
	Due to Lemma \ref{lemma:fulldim} and our assumption that $\Sigma_Z(\theta) \in \mathrm{int}(\Omega^Z)$, $T(S_n)$ satisfies the Le Cam regularity conditions required for the application of Wilk's Theorem \cite{wilks1938large}. This means $-2T(S_n)$ will be asymptotically distributed as a $\chi^2(m+q-p)$ random variable plus a fixed constant \textit{regardless of the true value of $\Sigma_{\Delta}$}, and thus implies that the event of a Type I error is independent of $\Sigma_{\Delta}$.
\end{proof}

This is a useful result because it implies that, in the proper regime, our test can come arbitrarily close to satisfying the initial goal of remaining robust to some uncertainty in the distribution of the measurement noise. However, $\Omega^{-1}$ is a non-convex set, and so the computation of $T(S_n)$ is difficult. To this end, we propose the approximate test statistic
\begin{equation}\tag{CRDW}
\label{eq:test4}
\begin{array}{rrcl}
\bar{T}(S_n)  = \min&\multicolumn{3}{l}{\Lcal(S_n,V)}\\
\mathrm{s.t.}&\sum_{k=1}^p\theta_k\bar{\Sigma}^{-1}_k&\succeq&V,\\
&\begin{bmatrix} V & I \\ I & \sum_{k=1}^p\theta_k\bar{\Sigma}_k\end{bmatrix}&\succeq&0,\\
&\mathbf{1}^{\textsf{T}}\theta^{\vphantom{\textsf{T}}}&=&1,\\
&\theta&\ge&\mathbf{0}.
\end{array}
\end{equation}

\begin{lemma}
	\label{lemma:feasreg}
	For any $V\in\Omega^{-1}$, there exists a $\theta\in\R^p$ such that $(V,\theta)$ is a feasible solution to the optimization problem defining test (\ref{eq:test4}).
\end{lemma}

\begin{proof}
First observe that any $V\in\Omega^{-1}$ can be written as $V = (\sum_{k=1}^p\theta_k\bar{\Sigma}_k)^{-1}$ for some nonzero $\theta$ such that $\mathbf{1}^{\textsf{T}}\theta^{\vphantom{\textsf{T}}}=1$. Thus, it holds trivially that
\begin{equation}
\label{eq:trivial}
\textstyle(\sum_{k=1}^p\theta_k\bar{\Sigma}_k)^{-1}\succeq V\succeq(\sum_{k=1}^p\theta_k\bar{\Sigma}_k)^{-1}
\end{equation}
The right-most constraint in (\ref{eq:trivial}) can be restated using the Schur complement, and this reformulation is exact. Since $\sum_{k=1}^p\theta_k\bar{\Sigma}_k\succeq0$, the Schur complement implies the second constraint in (\ref{eq:test4}) is equivalent to $V-(\sum_{k=1}^p\theta_k\bar{\Sigma}_k)^{-1}\succeq0$. 
	
	The first constraint in (\ref{eq:test4}) follows from the convexity of the matrix inverse for positive semidefinite matrices: Letting $X(\tau)=(1-\tau)X_1+\tau X_2$ for positive definite $n\times n$ matrices $X_1,X_2$ and $0\le\tau\le1$, we have $\frac{\nabla^2}{\nabla \tau^2}X(\tau)^{-1} = 2X^{-1}(\tau)X'(\tau)X^{-1}(\tau)X'(\tau)X^{-1}(\tau)$. For any $a\in\R^n$, the function $\phi_a(\tau)=a^{\textsf{T}}X^{-1}(\tau)a^{\vphantom{\textsf{T}}}$ will have second derivative $\phi_a''(\tau) = 2a^{\textsf{T}}X^{-1}(\tau)X'(\tau)X^{-1}(\tau)X'(\tau)X^{-1}(\tau)a^{\vphantom{\textsf{T}}}\ge0$ due to the positive-semidefiniteness of $X(\tau)^{-1}$, so $(1-\tau)\phi_a(0)+\tau\phi_a(1)\ge\phi_a(\tau)$. Since this holds for any $a$, we have that
\begin{equation}
\label{eq:invconvexity}
\textstyle\sum_{k=1}^p\theta_k\bar{\Sigma}_k^{-1}\succeq(\sum_{k=1}^p\theta_k\bar{\Sigma}_k)^{-1}.
\end{equation}
	
	\noindent The first constraint in (\ref{eq:test4}) follows from (\ref{eq:trivial}) and (\ref{eq:invconvexity}).
\end{proof}

\begin{remark}
\label{rem:fir}
It was shown in \cite{hespanhol2017dynamic} that test (\ref{eq:test2}) ensures $\alpha=0$ in any attack such that it holds true. In that case, we have
\begin{multline}
\label{eq:looseupperbound}
\textstyle\aslim_{N\rightarrow\infty}\frac{1}{N}\sum_{n=0}^{N-1}(C\hat{x}_n-y_n)^{\vphantom{\textsf{T}}}(C\hat{x}_n-y_n)^{\textsf{T}}\\
\textstyle	= C^{\vphantom{\textsf{T}}}\Sigma^{\vphantom{\textsf{T}}}_{\Delta}(\theta)C^{\textsf{T}} + \Sigma_Z(\theta) + \\
\textstyle\Sigma_S + \aslim_{N\rightarrow\infty}\frac{1}{N}\sum_{n=0}^{N-1}C^{\vphantom{\textsf{T}}}\eta_n^{\vphantom{\textsf{T}}}\eta_n^{\textsf{T}}C^{\textsf{T}},
\end{multline}
	since the Schur stability of $A+BK$ implies that any effect of $x_0$ and $\eta_0$ are reduced to zero asymptotically. Since $\Sigma_S$ and $\aslim_{N\rightarrow\infty}\frac{1}{N}\sum_{n=0}^{N-1}C^{\vphantom{\textsf{T}}}\eta_n^{\vphantom{\textsf{T}}}\eta_n^{\textsf{T}}C^{\textsf{T}}$ are both positive semidefinite, meaning that
\begin{multline}
\textstyle\aslim\frac{1}{N}\sum_{n=0}^{N-1}(C\hat{x}_n-y_n)^{\vphantom{\textsf{T}}}(C\hat{x}_n-y_n)^{\textsf{T}}\succeq \\
\textstyle C^{\vphantom{\textsf{T}}}\Sigma^{\vphantom{\textsf{T}}}_{\Delta}(\theta)C^{\textsf{T}} + \Sigma_Z(\theta).
\end{multline}
	
	Inverting both sides of this implies that, in the case that $\Sigma_S+ \aslim_{N\rightarrow\infty}\frac{1}{N}\sum_{n=0}^{N-1}C^{\vphantom{\textsf{T}}}\eta_n^{\vphantom{\textsf{T}}}\eta_n^{\textsf{T}}C^{\textsf{T}}\ne0$, we can generally expect that $S_n^{-1}\preceq (C^{\vphantom{\textsf{T}}}\Sigma^{\vphantom{\textsf{T}}}_{\Delta}(\theta)C^{\textsf{T}} + \Sigma_Z(\theta))^{-1}\in\Omega^{-1}$. The takeaway is that the looseness of the upper bound (\ref{eq:invconvexity}) should not greatly decrease the power of the modified test in the presence of test (\ref{eq:test2}), as the tight lower bound is more germane to situations where the system is actually being attacked.
\end{remark}

\begin{remark}
	If the dimension $m+q$ is large, then the optimization (\ref{eq:test4}) may be expensive to solve from scratch each time. Furthermore, $S_n$ will likely not change drastically between runs when $\ell$ is large. So, lighter-weight first-order methods such as ADMM can be used instead \cite{wen2010alternating}. These generally take longer to converge to high levels of accuracy, but have the advantage of being able to be readily warm-started.
\end{remark}

\subsection{Slowly Varying Unknown Noise Covariance}
\label{sec:varying}

A key difference between this setting and that of the static distribution is that a shift in the observer noise covariance in one period can have impacts on $\Sigma_{\Delta}$ over the next few periods that do not easily fit into our previous representation of the $\Omega$. This is because it will take many steps before the covariance of $\delta_n$ approaches its asymptotic limit in $\Omega$. Thus, to accommodate a dynamically changing distribution of $z_n$, we must use an expansion of the set $\Omega$.

We modify our setup for this subsection. The true covariance of $\delta_n$ and $z_n$ are $\Sigma_{\Delta_{n}}$ and $\Sigma_{Z_{n}}$, respectively. Let $\Psi_{n} = \Sigma_{Z_n} - \Sigma_{Z_{n-1}}$ and $\Phi_n^j = (A+LC)^jL\Psi_{n}L^{\textsf{T}}{(A+LC)^j}^{\textsf{T}}$. Note that all $\Sigma_{Z_n}$ are still assumed to be in $\Omega^Z$. Finally, we make some additional assumptions. Since the spectral radius of $A+LC$ is less than one, there exists some induced norm (denote this $\|\cdot\|$) such that $\|A+LC\|<1$ \cite{horn2012matrix}. We assume $\theta$ changes every step but $\Sigma_{Z_0}\in\Omega$ and all $\Psi_{n}$ satisfy $\|\Psi_{n}\|\le\xi$ for some known value of $\xi>0$. We also assume the system starts at steady state in the sense $\Sigma_{\Delta_{0}}=(A+LC)\Sigma_{\Delta_{0}}(A+LC)^{\textsf{T}}+\Sigma_W+L\Sigma_{Z_{0}}L^{\textsf{T}}$. Under these assumptions we have:

\begin{lemma}
\label{lemma:bound}
Let $\varepsilon\in\mathbb{R}$ be defined as
\begin{equation}
\label{eq:maxerror1}
\varepsilon=\frac{\xi\|C\|^2\|L\|^2\|A+LC\|^2\sqrt{m}}{\left(1-\|A+LC\|^2\right)^2}
\end{equation}
Then $C^{\vphantom{\textsf{T}}}\Sigma_{\Delta_{n}}^{\vphantom{\textsf{T}}}C^{\textsf{T}}+\Sigma_{Z_{n}}^{\vphantom{\textsf{T}}}\in\Omega\oplus\{E: -\varepsilon I\preceq E\preceq \varepsilon I\}$, where $\oplus$ is the Minkowski sum for all $n$.
\end{lemma}


\begin{proof}
	Let $\Omega_{m\times m}$ be the set of $m\times m$ upper-left submatrices of elements of $\Omega$, associated with $C^{\vphantom{\textsf{T}}}\Sigma_{\Delta}(\theta)^{\vphantom{\textsf{T}}} C^{\textsf{T}}+\Sigma_Z(\theta)$ terms. We start by noting that 
    \begin{equation}
        \begin{aligned}
        \Sigma_{\Delta_{1}} =& (A+LC)\Sigma_{\Delta_{0}}(A+LC)^{\textsf{T}}+\Sigma_W+L\Sigma_{Z_{1}}L^{\textsf{T}}\\
        =& \Sigma_{\Delta_{0}} + \Phi_{1}^{0}.
        \end{aligned}
    \end{equation}
    Similarly, we can see that $\Sigma_{\Delta_{2}}=\Sigma_{\Delta_{0}}+L\left(\Psi_{0}+\Psi_{1}\right)L^{\textsf{T}}+\Phi_{0}^{1}=\Sigma_{\Delta_{0}}+\Phi_{2}^{0}+\Phi_{1}^{0}+\Phi_{1}^{1}$. Continuing this recursion relation leads to the fact that 
    \begin{equation}
    \label{eq:expansion}
    \textstyle\Sigma_{\Delta_{n}} = \Sigma_{\Delta_{0}} + \sum_{i=0}^{n-1}\sum_{j=0}^{i}\Phi_{n-i}^{j}.
    \end{equation}
	Due to the Schur stability of $A+LC$, the following limit exists, and can be represented as in Lemma \ref{lemma:sigmadelta}.
	\begin{equation}
	\textstyle\Sigma_{\Delta_{\infty}^{k'}}=\lim_{k\rightarrow\infty}\big(\Sigma_{\Delta_{0}}+\sum_{i=k-k'}^{k-1}\sum_{j=0}^{i}\Phi_{k-i}^j\big)
	\end{equation}
	Note that $\Sigma_{\Delta_{\infty}^{k'}}$ is the steady state that $\Sigma_{\Delta_{n}}$ would ultimately reach if $\theta$ (and therefore $\Sigma_{Z_n}$ does not shift after step $k'$; thus, it solves (\ref{eq:sigmadelta}) for $\Sigma_{Z_{k'}}$ and exists in $\Omega_{m\times m}$. Denote $\Upsilon_{i}=\Sigma_{\Delta_{\infty}^{i}}-\Sigma_{\Delta_{\infty}^{i-1}}$. Then,
	\begin{equation}
	\begin{aligned}
	\label{eq:summation}
	\textstyle\Sigma_{\Delta_{n}}&=\textstyle\lim_{k\rightarrow\infty}\big(\Sigma_{\Delta_{0}}+\sum_{i=k-n}^{k-1}\sum_{j=0}^{i-k+n}\Phi_{k-i}^j\big)\\
	&\textstyle=\Sigma_{\Delta_{\infty}^{n}}-\lim_{k\rightarrow\infty}\big(\sum_{i=k-n}^{k-1}\big(\sum_{j=i-k+n+1}^{i}\Phi_{k-i}^j\big)\big)\\
	&\textstyle=\Sigma_{\Delta_{\infty}^{n}}-\sum_{i=1}^{n}(A+LC)^{n-i+1}\Upsilon_{i}{(A+LC)^{n-i+1}}^{\textsf{T}}
	\end{aligned}
	\end{equation}
	Note that the term in the limit in the first equality is a constant in $k$ due to a simple re-indexing of (\ref{eq:expansion}). This is convenient because we can now break $\Sigma_{\Delta_{n}}$ into an element known to be in $\Omega_{m\times m}$ and an error term. Our goal is now to choose $\varepsilon$ large enough to bound
	\begin{equation}
	\label{eq:error}
	\min_{P\in\Omega_{m\times m}}\big\|C^{\vphantom{\textsf{T}}}\Sigma_{\Delta_{n}}^{\vphantom{\textsf{T}}}C^{\textsf{T}}+\Sigma_{Z_{n}}-P\big\|_2,
	\end{equation}
    over all paths that $\Sigma_{Z_n}$ can take. An easy bound on the minimization is to simply set $P=C^{\vphantom{\textsf{T}}}\Sigma_{\Delta_{\infty}^{n}}^{\vphantom{\textsf{T}}}C^{\textsf{T}}+\Sigma_{Z_{n}}$. Then, $\varepsilon$ only needs to exceed
	\begin{equation}
	\textstyle\big\|\sum_{i=1}^{n}C(A+LC)^{n-i+1}\Upsilon_{i}{(A+LC)^{n-i+1}}^{\textsf{T}}C^{\textsf{T}}\big\|_2
	\end{equation}
    By sub-multiplicativity of induced norms,
	\begin{equation}
	\label{eq:upsilonbound}
	\begin{aligned}
	\|\Upsilon_{i}\| &= \textstyle\big\|\sum_{j=0}^{\infty}\Phi_{i}^j\big\|\le\sum_{j=0}^{\infty}\|(A+LC)\|^{2j}\|L\|\|\Psi_{n'+k_i}\|\\
	&=\xi\|L\|^2\big(1-\|A+LC\|^2\big)^{-1}
	\end{aligned}
	\end{equation}
	Finally, using the fact that $\|\cdot\|_2\le\sqrt{m}\|\cdot\|$ \cite{feng2003equivalence} and applying (\ref{eq:upsilonbound}) to the error term from (\ref{eq:error}) yields the desired result.
\end{proof}

\begin{remark}
Due to the topological equivalence of induced norms, the dependence of our choice of norm $\|\cdot\|$ on $A+LC$ can only affect the value of $\xi$ required by a constant $\sqrt{m}$.
\end{remark}

\begin{corollary}
If $\|A+LC\|_2<1$, then the statement in Lemma \ref{lemma:bound} holds for $\|\cdot\|=\|\cdot\|_2$ and
\begin{equation}
\label{eq:maxerror2}
\varepsilon=\frac{\xi\|C\|_2^2\|L\|_2^2\|A+LC\|_2^2}{\left(1-\|A+LC\|_2^2\right)^2}
\end{equation}
\end{corollary}

\begin{proof}
The proof of this result is almost identical to the proof of the previous lemma with the only changes that $\|\cdot\| = \|\cdot\|_2$ and that we stop after applying (\ref{eq:upsilonbound}) to (\ref{eq:error}).
\end{proof}

\begin{figure*}[t]
	\begin{center}
    	\begin{subfigure}[t]{0.9\linewidth}
        	\makebox[\textwidth]{\input{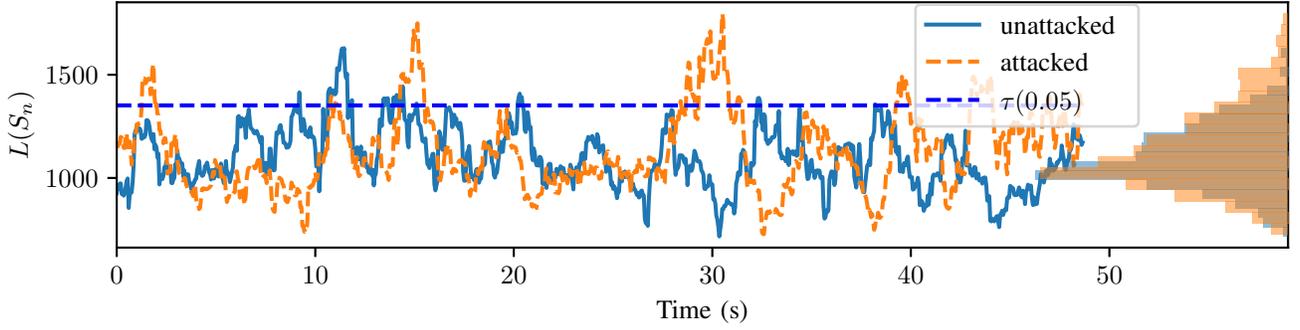}}
        	\caption{Test statistic (\ref{eq:oldstatistic})}
        	\label{fig:fixed_orig}
    	\end{subfigure}
    	\begin{subfigure}[t]{0.9\linewidth}
        	\makebox[\textwidth]{\input{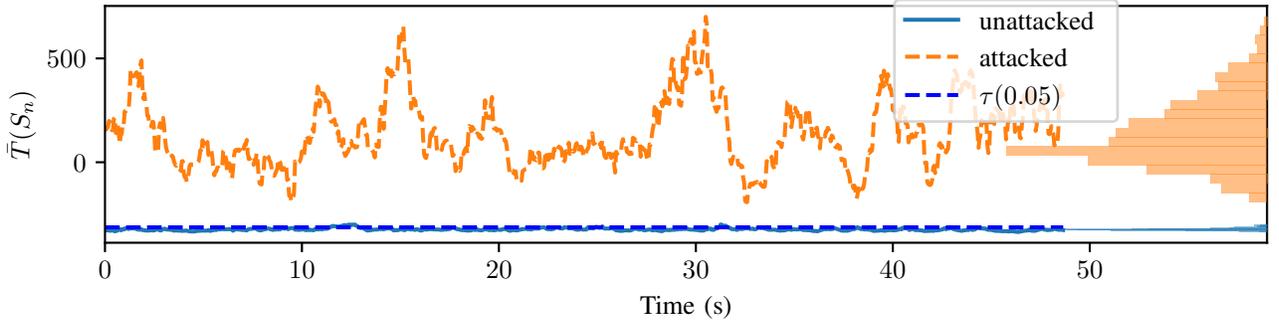}}
        	\caption{Test statistic (\ref{eq:test4})}
        	\label{fig:fixed_new}
    	\end{subfigure}
    	\caption{The evolution and histogram of test statistics (\ref{eq:oldstatistic}) and (\ref{eq:test4}) on the attacked and unattacked systems where $\Sigma_Z$ is fixed, but unknown to the tester. In this case, the nonrobust test statistic (\ref{eq:oldstatistic}) is unable to clearly distinguish the attacked from the unattacked system, whereas the new test statistic (\ref{eq:test4}) can.}
    	\label{fig:fixed_cov}
	\end{center}
\end{figure*}

With this $\varepsilon$, it is straightforward to extend the previous test statistic (\ref{eq:test4}) to this new expansion of $\Omega$ as long as $\bar{\Sigma}_k-\varepsilon I$ remains positive definite for all $k$. In this case, we may define our new test statistic as
\begin{equation}\tag{CRDW*}
\label{eq:test5}
\begin{array}{rrcl}
\underline{T}(S_n)  = \min&\multicolumn{3}{l}{\Lcal(S_n,V)}\\
\mathrm{s.t.}&\sum_{k=1}^p\theta_k\left(\bar{\Sigma}_k-\varepsilon I\right)^{-1}&\succeq&V,\\
&\begin{bmatrix} V & I \\ I & \varepsilon I+\sum_{k=1}^p\theta_k\bar{\Sigma}_k\end{bmatrix}&\succeq&0,\\
&\mathbf{1}^{\textsf{T}}\theta^{\vphantom{\textsf{T}}}&=&1,\\
&\theta&\ge&\mathbf{0}.
\end{array}
\end{equation}
\begin{remark}
If there is some $k$ so $\bar{\Sigma}_k-\varepsilon I$ is not positive definite, then the first constraint above is not well-defined. Recalling that $V$ is a surrogate for $\left(C^{\vphantom{\textsf{T}}}\Sigma_{\Delta_n}^{\vphantom{\textsf{T}}}C^{\textsf{T}}+\Sigma_{Z_{n}}\right)^{-1}$, we note $V$ trivially satisfies $\Sigma_{Z_n}^{-1}\succeq V$. Thus in this problematic case, we may replace the $\left(\bar{\Sigma}_k-\varepsilon I\right)$ in the first constraint with $\Sigma_{z,k}$, for all $k$. This issue is unlikely to be of practical concern for the same reasons discussed in Remark \ref{rem:fir} regarding the relaxation of the set $\Omega$. Specifically, the structure of the attacks makes it unlikely that the first constraint in (\ref{eq:test5}) would be binding in any case.
\end{remark}

\begin{figure*}[!t]
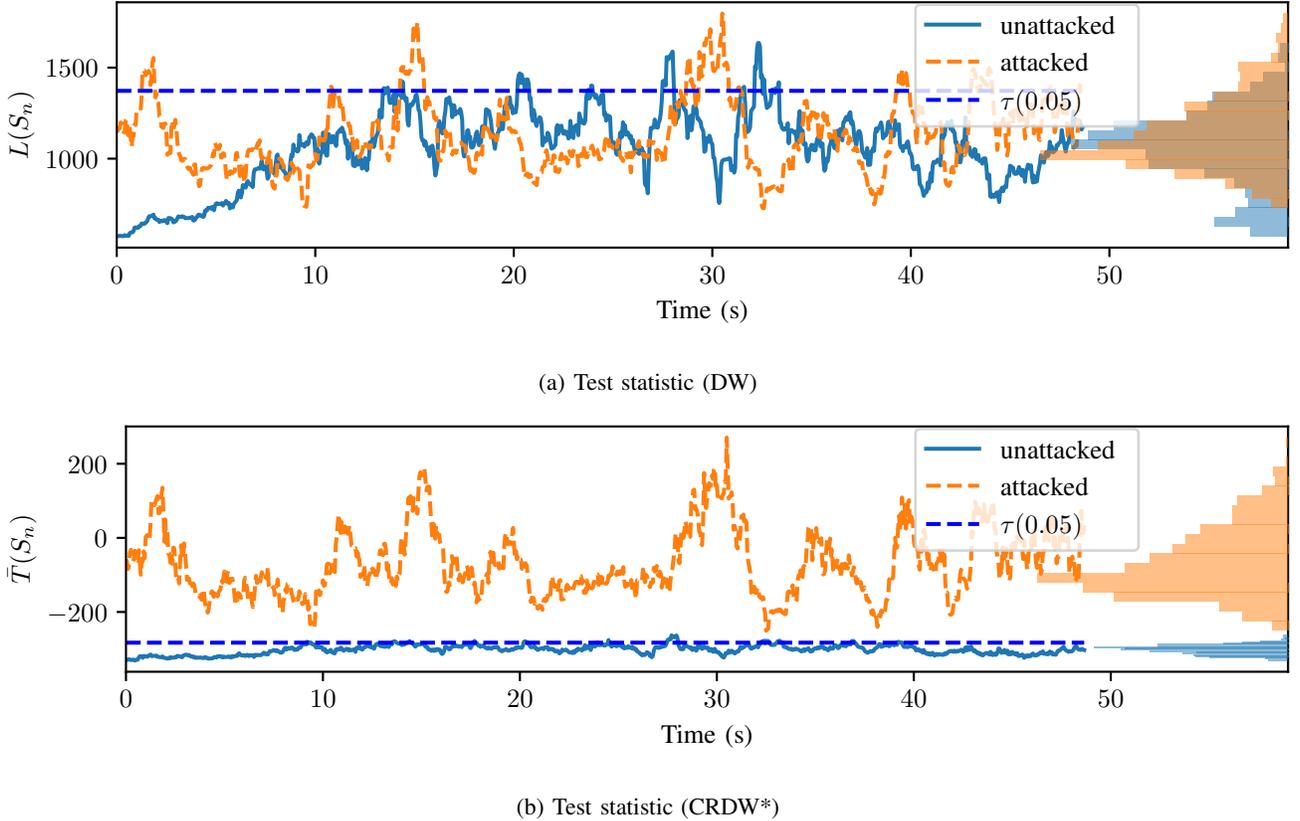

	\begin{center}
    	\begin{subfigure}[t]{0.9\textwidth}
        	\makebox[\textwidth]{\input{orig_test_varying_noise.pgf}}
        	\caption{Test statistic (\ref{eq:oldstatistic})}
        	\label{fig:varying_orig}
    	\end{subfigure}
    	\begin{subfigure}[t]{0.9\linewidth}
        	\makebox[\textwidth]{\input{new_test_varying_noise.pgf}}
        	\caption{Test statistic (\ref{eq:test5})}
        	\label{fig:varying_new}
    	\end{subfigure}
    	\caption{The evolution and histogram of test statistics (\ref{eq:oldstatistic}) and (\ref{eq:test5}) on the attacked and unattacked systems where $\Sigma_Z$ varies as described, again unknown to the tester. Note the robust statistic (\ref{eq:test5}) takes distinctly higher for the attacked values over almost the entire 1000 iterations than in the unattacked system, while the nonrobust statistic (\ref{eq:oldstatistic}) is again unable to clearly distinguish the two.}
    	\label{fig:varying_cov}
	\end{center}
\end{figure*}

\section{Empirical Results}
\label{sec:results}

In this section, we present simulation results that showcase the strength of our method when compared with the original test statistic (\ref{eq:oldstatistic}). We present results for both the case where the noise distribution is fixed but unknown, and for the case where the noise covariance is unknown and slowly-varying.

We use the standard model for simulation of an autonomous vehicle in \cite{turri2013linear}, where the error kinematics of lane keeping and speed control is given by $x^{\textsf{T}} = \begin{bmatrix} \psi & y & s & \gamma & v \end{bmatrix}$ and $u^{\textsf{T}} = \begin{bmatrix} r&a \end{bmatrix}$. Here, $\psi$ is heading error, $y$ is lateral error, $s$ is trajectory distance, $\gamma$ is vehicle angle, $v$ is vehicle velocity, $r$ is steering, and $a$ is acceleration. We linearize and initialize with a straight trajectory and constant velocity $v_{0} = 10$. We then  performed exact discretization with sampling period $t_{s} = 0.05$. This yields the system dynamics

\begin{equation}
A = \begin{bmatrix}
1 & 0 & 0 & \frac{1}{10} & 0 \\
\frac{1}{2} & 1 & 0 & \frac{1}{40} & 0 \\
0 & 0 & 1 & 0 & \frac{1}{2} \\
0 & 0 & 0 & 1 & 0 \\
0 & 0 & 0 & 0 & 1
\end{bmatrix},\;\;B = \begin{bmatrix}
\frac{1}{400} & 0  \\
\frac{1}{2400} & 0  \\
0 & \frac{1}{800}  \\
\frac{1}{20} & 0 \\
0 & \frac{1}{20} 
\end{bmatrix}
\end{equation}

\noindent with $C = 
\begin{bmatrix}
I & 0
\end{bmatrix}
\in\R^{3\times5}$. We use process noise covariance $\Sigma_W = 10^{-8}\times I$. 

All tests use dynamic watermarking with variance $\Sigma_E = \frac{1}{2}I$, and $K$ and $L$ were chosen to stabilize the system without an attack. We conduct four simulations: attacked and non-attacked systems where the measurement noise covariance is fixed, and attacked and non-attacked systems where the measurement noise covariance is allowed to vary. We ran all four simulations for 1000 iterations, or 50 seconds. In all cases, we compare the test metrics using the hypothesis test described in (\ref{eq:test3}), where the measurement noise covariance is assumed to be $10^{-5}\times I$. When simulating the attacked system, we choose an attacker with $\alpha = -1$, $\eta_{0} = 0$, $\Sigma_O = 10^{-8}\times I$, and $\Sigma_S = 10^{-8}\times I$.

\subsection{Fixed Covariance}

We first show our test outperforms in the case where the true measurement noise covariance matrix is fixed but unknown to the tester. In our simulations, the true noise covariance is $\Sigma_Z = 10^{-5}\times\textrm{diag}\{0.18,30,0.18\}$. In all tests, $\Omega^Z$ is described by the $p=4$ extreme points: $\Sigma_{Z,1} = 10^{-6}\times\textrm{diag}\{300,1.8,1.8\}$, $\Sigma_{Z,2} = 10^{-6}\times\textrm{diag}\{1.8,300,1.8\}$, $\Sigma_{Z,3} = 10^{-6}\times\textrm{diag}\{9,9,12\}$, $\Sigma_{Z,4} = 10^{-6}\times\textrm{diag}\{9,9,9\}$. Both the true measurement noise covariance and that incorrectly assumed in test statistic (\ref{eq:oldstatistic}) are in the resulting set. The simulation is run for 1000 steps.

\Cref{fig:fixed_cov} shows the efficacy of our method under this new uncertainty. If test detection is consistent, the negative log likelihood values should be lower under regular conditions, and higher when the model is attacked. In particular, the nonrobust test statistic (\ref{eq:oldstatistic}) is shown in Fig. \ref{fig:fixed_orig} to be wholly unable to distinguish an attacked system from an unattacked system when its assumption on the measurement noise covariation is violated, while Fig. \ref{fig:fixed_new} shows the robust test statistic (\ref{eq:test4}) to be able to do so.

\subsection{Varying Covariance}

Unattacked and attacked simulations were also conducted with a measurement noise distribution that was allowed to vary. We set $\tau=1$ and $\xi=0.00002$, implying $\varepsilon =  7.205\times10^{-6}$. The true measurement noise is initialized at $\Sigma_{Z_0} = 10^{-5}\times\textrm{diag}\{0.9,0.9,1.2\}$. This shifts linearly over the course of 250 iterations to a new value of $\Sigma_{Z_{250}} = 10^{-5}\times\textrm{diag}\{15,15,0.18\}$, at which point it changes direction to shift linearly over 250 iterations to a value of $\Sigma_{Z_{500}} = 10^{-5}\times\textrm{diag}\{30,0.18,0.18\}$. The measurement noise covariance stays at this value for 150 iterations. It then shifts linearly over 200 iterations to a terminal value of $\Sigma_{Z_{850}} = 10^{-5}\times\textrm{diag}\{0.18,30,0.18\}$, which it takes for another 150 iterations before the simulation is terminated. The results for both the nonrobust and robust tests are shown in Fig. \ref{fig:varying_cov}. As in the fixed covariance case, our test is able to distinguish between the attacked and unattacked systems better and more consistently than the nonrobust test that requires unsatisfied assumptions.

\section{Conclusion}
\label{sec:conclusion}

We developed covariance-robust dynamic watermarking tests for detecting sensor attacks on LTI systems in the presence of uncertainty about the measurement noise covariance. We considered cases where the covariance of measurement noise is unknown and either fixed or slowly-varying, and we required our test to be ``fair" with respect to all possible values of the covariance in that it not be more or less powerful for some covariances over others. These reflect real-world needs that will increase as 5G is deployed, because there will be an increase in the deployment of smart CPS systems. In such systems, an ``unfair" test can translate to disparate impact across different users in different environments, which is a problem of algorithmic bias. Future research includes studying how dynamic watermarking can be adapted to other system uncertainties.

%
%
%





\bibliographystyle{IEEEtran}
\bibliography{IEEEabrv,fdwm}

\end{document}